\pdfoutput=1
\documentclass{article}

\PassOptionsToPackage{numbers}{natbib}

\usepackage[preprint]{neurips_2023}

\usepackage[utf8]{inputenc} 
\usepackage[T1]{fontenc}    
\usepackage{hyperref}       
\usepackage{url}            
\usepackage{booktabs}       
\usepackage{amsfonts}       
\usepackage{nicefrac}       
\usepackage{microtype}      
\usepackage{xcolor}         
\usepackage{xspace}         

\usepackage{latexsym, amsmath, amssymb, amsthm}
\usepackage{enumerate}
\usepackage{enumitem}
\usepackage[ruled, vlined, linesnumbered]{algorithm2e}
\usepackage[nameinlink, noabbrev, capitalise]{cleveref}
\usepackage{mathtools}
\usepackage{multicol}
\usepackage{caption}
\usepackage{subcaption}

\newcommand{\sol}{\textup{\texttt{VISER}}\xspace}
\newcommand{\mpsol}{\textup{\texttt{MPVISER}}\xspace}
\DeclareMathOperator*{\argmax}{arg\,max}

\newcommand{\Real}{\mathbb{R}}
\newcommand{\brac}[1]{\left[ #1 \right]}
\newcommand{\set}[1]{\left\{ #1 \right\}}
\DeclareMathOperator*{\E}{E}
\DeclareMathOperator*{\BR}{BR}
\newcommand{\NE}{\textup{NE}}

\theoremstyle{plain}
\newtheorem{theorem}{Theorem}
\newtheorem{lemma}{Lemma}
\newtheorem{corollary}{Corollary}

\newtheorem{observation}{Observation}
\newtheorem{proposition}{Proposition}
\newtheorem{assumption}{Assumption}

\theoremstyle{definition}
\newtheorem{definition}{Definition}

\title{\sol: A Tractable Solution Concept for Games with Information Asymmetry}

\author{%
  Jeremy McMahan \\ 
  University of Wisconsin-Madison \\
  \texttt{jmcmahan@wisc.edu} \\
  \And
  Young Wu \\ 
  University of Wisconsin-Madison \\
  \texttt{yw@cs.wisc.edu} \\
  \And
  Yudong Chen \\ 
  University of Wisconsin-Madison \\
  \texttt{yudong.chen@wisc.edu} \\
  \And 
  Xiaojin Zhu \\ University of Wisconsin-Madison \\
  \texttt{jerryzhu@cs.wisc.edu} \\
  \And 
  Qiaomin Xie \\ University of Wisconsin-Madison \\
  \texttt{qiaomin.xie@wisc.edu} \\
}

\begin{document}

\maketitle

\begin{abstract}
Many real-world games suffer from information asymmetry: one player is only aware of their own payoffs while the other player has the full game information. Examples include the critical domain of security games and adversarial multi-agent reinforcement learning. Information asymmetry renders traditional solution concepts such as Strong Stackelberg Equilibrium (SSE) and Robust-Optimization Equilibrium (ROE) inoperative. We propose a novel solution concept called VISER (Victim Is Secure, Exploiter best-Responds). VISER enables an external observer to predict the outcome of such games. In particular, for security applications VISER allows the victim to better defend itself while characterizing the most damaging attacks available to the exploiter.  We show that each player’s VISER strategy can be computed independently in polynomial time using linear programming (LP). We extend VISER to its Markov Perfect counterpart in Markov games, which can also be solved efficiently by a series of LPs.
\end{abstract}


\section{Introduction}\label{sec: intro}

As multi-agent systems become increasingly distributed and complex, information asymmetry amongst agents becomes inevitable. 
We focus on two-player general-sum games (and by extension, Markov games) with the following information asymmetry: the first player knows the game except the second player's payoff function, while the second player has full information of the game.
Obviously, the second player can exploit its information advantage to optimize its own payoff, possibly at the expense of the first player. As such, we refer to the second player as the \emph{exploiter} and the first player as the \emph{victim}. 
For example, in security applications and adversarial MARL the victim often does not know the exploiter's true incentive or cost of taking each attack action~\citep{StackelbergAppsExst, gleave2021adversarial}. 
In contrast, the exploiter knows the victims payoff in addition to its own.
Given such a game with information asymmetry, can we predict what rational players will do?

\paragraph{Prior work.}
Incomplete information games were first studied through the framework of Bayesian games~\citep{BNE1, BNE2, BNE3} with corresponding solution concept being the Bayes Nash Equilibrium (BNE). Importantly, the Bayesian approach requires the players to have ``known unknowns'': an accurate distribution of the unknown parameters of the game. To address the high sensitivity BNE with respect to the player's beliefs~\citep{BNESensExample, BNESensitivity}, the work~\citep{ExPostEquil} introduced a more robust equilibrium concept called ex-post equilibrium, which is a NE under all possible realizations of the uncertain parameters. Then,~\citep{RobustDistr} studied another robust variation, where the players have multiple belief distributions and they compute the robust solution (maximin) over the worst possible scenario. 
More recently, the Bayesian approach has been extended to deal with type-uncertainty in Markov security games~\citep{BayesMarkov}. 



Going beyond the need for belief distributions,~\citep{RGT} introduced the notion of robust games where the uncertainty is defined by a deterministic set over parameters defining the shape of the payoff functions. The solution concept to a robust game is a robust-optimization equilibrium (ROE), which is found using robust optimization techniques with respect to the worst-case parameters in the uncertainty set. ROE is a distribution-free solution concept that was shown to be less conservative than ex-post equilibrium in that every ex-post equilibrium is an ROE~\citep{RGT}. Then,~\citep{GRGT} generalized ROE to allow a broader class of payoff functions and showed that ROE can be seen as an $\epsilon$-NE in the game without uncertainty. Recently,~\citep{SRGT} studied how sensitive ROE is with respect to the uncertainty levels of the players. However, both the BNE and ROE approach require non-trivial assumptions about the information structure: namely, an uncertainty parametrization or distributional assumption on the opponent's payoffs. Thus, they do not apply to our setting where the victim knows nothing about the exploiter's payoffs.


Besides the issue of requiring a parametrization or distribution of the opponent's payoff, computing a NE (BNE and ROE) of a general-sum game is PPAD-hard. Luckily, using the power of commitment, a  Strong Stackelberg Equilibrium (SSE), which yields the victim a higher payoff than in a NE, can be computed in polynomial-time for two-player games~\citep{Stackelberg, StackelbergCompute}. In fact, even for security games with exponentially large strategy spaces, a SSE can still be computed efficiently~\citep{StackCompSecurity}. As many real-world security games can satisfy the commitment assumption, the SSE has been the dominant solution concept for security applications~\citep{StackelbergSurvey}. For example, Stackelberg strategies are the heart of LAX's airport security ARMOR system~\citep{StackelbergLAX, StackelbergApplications} and are also used in the system IRIS used by the US Federal Air Marshals~\citep{StackelbergVsNash} with many other important applications~\citep{StackelbergAppsExst}. Recently, the Stackelberg approach has been extended to Markov security games~\citep{StackelbergMarkov}, but at the cost of even stronger assumptions on the information structure of the players. Information asymmetry has also been studied for SSE, but only when the asymmetry comes from either the exploiter being unable to fully see the commitment~\citep{UnobservableCommit} or the victim has more information~\citep{Signaling, TwoStage}, which is virtually the opposite of our setting. 

\paragraph{Our contributions.}
First, (i) we propose and study a new solution concept, \sol, which is appropriate for games with information asymmetry. 
\sol stands for \textit{\underline{v}ictim \underline{i}s \underline{s}ecure and the \underline{e}xploiter best-\underline{r}esponds}. 
\sol posits that a rational victim, lacking information on the exploiter, will necessarily choose from the set of its own maximin strategies.
The exploiter can reason the victim will do this, but does not know the particular maximin strategy chosen. Thus, a rational exploiter must best respond to the worst-case member of the victim maximin strategy set.
Importantly, a \sol solution always exists and each player can independently play its \sol strategy without the need for coordination or signaling. Also, if either player unilaterally deviates from a \sol solution, then they will achieve lower payoff in the worst-case, so \sol functions as an an equilibrium concept. Second, (ii) we show that each player can compute a \sol strategy independently in polynomial time by solving a linear program (LP). A key technical insight is that the exploiter can use LP duality to convert its robust, best-response optimization problem into a single LP. Third, (iii) we generalize our framework to Markov games (also known as stochastic games). Here, it is more natural to study Markov-perfect \sol solutions (\mpsol), which consist of Markovian policies, as standard in the theory of stochastic games~\citep{Markov}. We show a \mpsol solution is a \sol solution, it always exists, and each player can compute a \mpsol solution in polynomial time by solving a sequence of LPs via backward induction. 

\section{Preliminaries}\label{sec: preliminaries}

\paragraph{Bimatrix Games.} In a (finite) bimatrix game, two players compete simultaneously to maximize their own payoff. Suppose that the first player, the victim, has $n$ pure strategies and the second player, the exploiter, has $m$ pure strategies. Let $A \in \Real^{n \times m}$ and $B \in \Real^{n \times m}$ denote the payoff matrices for the victim and exploiter, respectively. We may represent a pure strategy by a one-hot vector, so $e_i \in \Real^n$ corresponds to victim's strategy $i$ and $e_j \in \Real^m$ the exploiter's strategy $j$. Let $ \Delta(k) := \big\{ s \in [0,1]^k \mid \sum_{i = 1}^k s_i = 1 \big\}$ denote the set of mixed strategies, where choosing $s \in  \Delta(k)$  corresponds to playing $e_i$ with probability $s_i$. For a pair mixed strategies $x \in \Delta(n)$ and $y \in \Delta(m)$, the expected payoffs to the victim and exploiter are $x^\top A y$ and $x^\top B y$, respectively.


\paragraph{Nash Equilibrium and Security Strategies.} Solutions to games manifest as equilibrium concepts, among which the most famous is the \emph{Nash Equilibrium} (NE)~\citep{Nash}. An NE of a bimatrix game is a pair of strategies $(x^*,y^*) \in \Delta(n)\times\Delta(m)$ satisfying 
\begin{equation*}\label{equ: nash}
x^* \in \argmax_{x \in \Delta(n)} x^\top A y^*, \;\;\text{ and }\;\; y^* \in \argmax_{y \in \Delta(m)} {x^*}^\top B y.
\end{equation*}
Another solution concept is a \emph{maximin strategy} or \emph{security strategy}, which is a pair $(x^*,y^*)$ given by 
\begin{equation}
\label{equ: maximin}
x^* \in \argmax_{x \in \Delta(n)} \min_{y \in \Delta(m)} x^\top A y, \;\;\text{ and }\;\; y^* \in \argmax_{y \in \Delta(m)} \min_{x \in \Delta(n)} x^\top B y.
\end{equation}
In a \emph{zero-sum} game ($B = -A$), the Minimax Theorem~\citep{maximin} implies $(x^*,y^*)$ is a NE if and only if it is a maximin strategy pair. Note that a game may have multiple NEs and maximin strategies. However, in zero-sum games, each player receives the same expected payoff in every NE, which we denote by $p_v^{NE}$ and $p_e^{NE}$ respectively. 

\paragraph{Markov Games.} 
A finite-horizon \emph{Markov game}~\citep{Markov} is defined by a tuple $G = (S,A,R,P,H,\mu)$ with
state space $S$,
joint action space $A = A_v \times A_e = [n] \times [m]$ ($[i] := \set{1, \ldots, i}$),
joint reward function $R$, 
transition function $P$,
horizon $H$,
and initial state distribution $\mu$. 
We denote by $\pi = \{\pi_h(s) \in \Delta(n)\}_{h,s}$ a Markovian policy for the victim and by $\nu = \{\nu_h(s) \in \Delta(m)\}_{h,s}$ a Markovian policy for the exploiter. Let $\Pi$ denote the set of all Markovian policies for the victim and $N$ denote the set of all Markovian policies for the exploiter. For $u\in\{v,e\}$ (victim and exploiter), the value/payoff received by player $u$ under $(\pi,\nu)$ is the expected cumulative rewards over $H$ time steps: $V_u^{\pi,\nu} \equiv V_{1u}^{\pi,\nu}  := \E_{\pi,\nu}\brac{\sum_{h = 1}^H \pi_{h}(s_h)^\top R_{hu}(s_h)\nu_{h}(s_h)}$.
Similarly define the stage value, $V_{hu}^{\pi^*,\nu^*}(s)$, for each $h\in[H]$ by summing over steps $h$ through $H$. 


Equilibrium concepts can be defined for a Markov Game by viewing it as a (very large) bimatrix game with payoff matrices $(V_v^{\pi,\nu})_{\pi,\nu}$ and $  (V_e^{\pi,\nu})_{\pi,\nu}$. To avoid this complexity blowup, many works focus on \emph{Markov Perfect Equilibrum} (MPE), which requires the stricter property that a policy pair is an equilibrium at \emph{every} stage game, not just at stage $h=1$. 
Formally, $(\pi^*, \nu^*)$ is a MPE if, for all $(h,s)\in [H]\times S$, 
\begin{equation*}\label{equ: mpe}
    V_{hv}^{\pi^*,\nu^*}(s) = \max_{\pi \in \Pi} V^{\pi,\nu^*}_{hv}(s)  \;\;\text{ and }\;\; V_{he}^{\pi^*,\nu^*}(s) = \max_{\nu \in N} V^{\pi^*,\nu}_{he}(s).
\end{equation*}

\section{\sol: A New Solution Concept}\label{sec: sol}


Suppose the victim and exploiter compete in a bimatrix game $(A,B)$ with information asymmetry. The victim knows only $A$, but the exploiter knows
both $A$ and $B$. The exploiter naturally wishes to leverage its information advantage: it seeks to derive a strategy $y^*$ that achieves a higher value than it would from only knowing $B$. On the other hand, the victim is aware of its information disadvantage and thus simply wants to guarantee a good value regardless of the exploiter's strategy.


\paragraph{Rationality.} As shown in \citep{UncertaintyAversion}, real players are averse to uncertainty. This is reinforced by the fact if a player guesses some belief incorrectly, it risks achieving arbitrarily poor value~\citep{BNESensExample}. Consequently, as argued by worst-case approaches in game theory~\citep{RGT}, it is rational for such players to use a strategy that optimizes its worst-case scenario since it cannot formulate any accurate beliefs of the other player's payoff. We assume the players follow this reasoning, which we call \emph{worst-case rationality}.

\begin{assumption}(Worst-Case Rationality)\label{assump: rationality}
    Both players seek to optimize their worst-case payoff with respect to the information available to them, and this fact is common knowledge to both players.
\end{assumption}


\paragraph{Victim's Perspective.} Absent any information about the exploiter's strategy or payoff, the safest option for the victim is to play a security strategy \cref{equ: maximin}. A security strategy guarantees the victim a baseline payoff regardless of how the exploiter uses its information advantage to pick a strategy. 
And, importantly, a security strategy can be computed from knowledge of $A$ alone.
Equivalently, since the victim has no knowledge whatsoever about $B$, we can view them as \emph{assuming} that $B=-A$, i.e., the exploiter is adversarial, and then computing an NE of the zero-sum game $(A,-A)$.
\begin{observation}\label{assump: victim}(Victim Behavior)
    Under \cref{assump: rationality}, the victim will play:
    \begin{equation*}\label{equ: victim}
        x^* \in X^* := \argmax_{x \in \Delta(n)} \min_{y \in \Delta(m)} x^\top A y,
    \end{equation*}
which guarantees the victim an expected payoff of at least,
\begin{equation*}
p_v := \max_{x \in \Delta(n)} \min_{y \in \Delta(m)} x^\top A y.
\end{equation*}
\end{observation}



The intuition is that if the victim does not play some $x^* \in X^*$ it could receive \emph{arbitrarily worse} payoff in the worst case. To illustrate this point, consider the game in \cref{table: true-game}. 
If the victim were to optimistically, but incorrectly, believe the game was actually the one in \cref{table: fake-game}, then it would play action D as part of the NE (D, R). However, in the true game~\ref{table: true-game}, playing D results in the worst-case scenario of receiving $0$ payoff. In contrast, if the victim were to have played its unique security strategy, U, it would have received at least a payoff of $x$, which is arbitrarily larger than $0$. 


\begin{figure}
    \centering
    \begin{subfigure}[b]{.3\textwidth}
        \centering
        \begin{tabular}{c|c|c|}
          & L & R \\
         \hline
        U & (x,y) & (2x,0) \\
        \hline
        D & (0,y) & (3x,0) \\
        \hline
        \end{tabular}
        \caption{Victim's Example}
        \label{table: true-game}
    \end{subfigure}
    \vline
    \begin{subfigure}[b]{.3\textwidth}
        \centering
        \begin{tabular}{c|c|c|}
         & L & R \\
         \hline
        U & (x,0) & (2x,y) \\
        \hline
        D & (0,0) & (3x,y) \\
        \hline
        \end{tabular}
        \caption{Victim's Belief}
        \label{table: fake-game}
    \end{subfigure}
    \vline
    \begin{subfigure}[b]{.3\textwidth}
        \centering
        \begin{tabular}{c|c|c|}
         & L & R \\
         \hline
        U & (x,2y) & (x,-1)  \\
        \hline
        M & (x,y) & (x,-1) \\
        \hline
        D & (-1,-1) & (-1,0) \\
        \hline
        \end{tabular}
        \caption{Exploiter's Example}
        \label{table: exploiter-game}
    \end{subfigure}
    \caption{Example Games}
    \label{fig: examples}
\end{figure}

\paragraph{Exploiter's Perspective.} The exploiter could also compute its own security strategy, but this would not take advantage of its information advantage. The exploiter knows both $A$ and $B$ and can reason that under \cref{assump: rationality} the victim will play some $x^* \in X^*$ as indicated by \cref{assump: victim}. However, the exploiter is not completely omniscient. It does not know which $x^* \in X^*$ is chosen and so suffers from strategic uncertainty~\citep{Ambiguity}. Thus, the exploiter should assume the worst-case $x^* \in X^*$ is chosen by the victim and simply best-respond to it. Here, we generalize the best-response function to a function on a set $X^*$ by considering the worst-case best-response. 
\begin{observation}\label{assump: exploiter}(Exploiter Behavior)
    Under \cref{assump: rationality}, the exploiter will play:
    \begin{equation*}\label{equ: exploiter}
        y^* \in \BR(X^*) := \argmax_{y \in \Delta(m)} \min_{x^* \in X^*} x^{*^\top} B y,
    \end{equation*}
which guarantees the exploiter an expected payoff of at least,
\begin{equation*}
p_e := \max_{y \in \Delta(m)} \min_{x^* \in X^*} x^{*^\top} B y.
\end{equation*}
\end{observation}

Critically, the exploiter uses its information advantage to obtain a payoff higher than what can be achieved by an unrestricted best response $y\in\BR(\Delta(n))$, which consists of its security strategies as defined in \cref{equ: maximin}. 
In fact, playing a security strategy could lead to \emph{arbitrarily worse} payoff for the exploiter. Consider the game in \cref{table: exploiter-game}. It is clear that any mix of U and M yields a security strategy for the victim, which achieves the victim a payoff of $x$. The exploiter's best response in the worst case is L, which achieves it a payoff of at least $y$. In contrast, its unique security strategy is R, which achieves a payoff of at most $0$ in general and at most $-1$ if the victim plays rationally. Thus, playing the worst-case best-response leads to an arbitrarily higher payoff for the exploiter.


\paragraph{\sol Solution.} 
Under \cref{assump: rationality}, \cref{assump: victim} and \cref{assump: exploiter} imply a game with information asymmetry will play out as follows: the victim will play a security strategy from $X^*$ and the exploiter plays a worst-case best-response strategy from $\BR(X^*)$. This leads to our solution concept we call \sol for \emph{\underline{v}ictim \underline{i}s \underline{s}ecure and the \underline{e}xploiter best-\underline{r}esponds}.
\begin{definition}(\sol)
\begin{equation*}
    \sol := X^* \times \BR(X^*).
\end{equation*}
\end{definition}
Importantly, each player can compute its own part of a \sol pair given only the information available to it: no coordination or signaling schemes are needed. Also, a \sol solution always exists. 

\begin{theorem}\label{prop: sol}
    Under \cref{assump: rationality}, the \sol concept satisfies the following properties:
    \begin{itemize}
        \item A \sol solution always exists ($\sol \neq \varnothing)$,
        \item A \sol solution can be deployed in a distributed fashion, and
        \item Deviating from a \sol solution achieves each player a worse payoff in the worst case.
    \end{itemize}
    Thus, in a bimatrix game where the victim knows only $A$, but the exploiter knows both $A$ and $B$, under \cref{assump: rationality}, the players will choose a strategy pair $(x^*,y^*) \in \sol$.
\end{theorem}


\paragraph{Expected Payoffs.} We can characterize the payoffs achieved by both players under \sol. Let $(x^*, y^*) \in \sol$. We see the victim's payoff by \cref{assump: victim} is at least as good as the maximin value, $p_v$, but could be even higher depending on the choice of $y^*$. Similarly, the exploiter's payoff by \cref{assump: exploiter} is at least as good as the $X^*$-restricted maximin value, $p_e$, but could be even higher depending on the choice of $x^*$. When the game is zero-sum, the maximin theorem implies that $p_v = p_v^\NE$ is the unique NE payoff. Similarly, the exploiter best responds to $X^*$ and so best responds to some Nash strategy, which gives a Nash strategy for the exploiter. Hence, $p_e = p_e^\NE$ as well from the definition of Nash as mutual best-responses.




\begin{proposition}
    For any $(x^*,y^*) \in \sol$, 
    \begin{equation*}
        {x^*}^\top A y^* \geq p_v, \text{ and } {x^*}^\top B y^* \geq p_e.
    \end{equation*}
    If the game is zero-sum, then:
    \begin{equation*}
        {x^*}^\top A y^* = p_v^{\NE}, \text{ and } {x^*}^\top B y^* = p_e^{\NE}.
    \end{equation*}
\end{proposition}



    

\paragraph{Comparison to Stackelberg Strategies.} We can view \sol through the lens of Stackelberg strategies~\citep{Stackelberg, StackCompSecurity}. Our common knowledge assumption functionally announces to the exploiter that the victim will play from $X^*$. However, unlike typical Stackelberg strategies, rather than committing to a single mixed strategy, the victim here is committing to a set of strategies. Consequently, the exploiter now must best respond to an entire set rather than just a single strategy. Thus, \sol strategies can be seen as a generalization of Stackelberg strategies with set-wise commitments and set-wise best responses. In fact, our algorithms in \cref{sec: bimatrix} continue to work so long as the victim commits to a succinctly representable polytope of strategies.

\section{Linear Programs for Computing \sol in Bimatrix Games}\label{sec: bimatrix}

Next, we study the computational complexity of computing \sol solutions. We show both players can efficiently compute a \sol strategy using linear programming. For the victim, the LP is derived analogously to the LP for computing a NE. The exploiter's problem, however, is more challenging since it must best respond to $X^*$, a potentially infinite set of mixed strategies.\footnote{The polytope $X^*$ may have an exponential (in $n,m$) number of extreme points in the worst-case.} Despite its apparent difficulty, we use linear programming duality to derive an LP that computes the exploiter's optimal strategy.

\begin{figure}
    \centering
    \begin{subfigure}[b]{.48\textwidth}
        \centering
        \begin{align*}
        &\max_{x\in \Real^{n}, z\in \Real}\quad  \ z  \\
        \text{s.t. \quad}  &z \leq x^\top A e_j, \quad \forall j \in [m] \\
    	&      1^\top x = 1, \quad x \geq 0.
        \end{align*}
        \caption{Victim's LP}
        \label{equ: victim-LP}
    \end{subfigure}
    \vline
    \begin{subfigure}[b]{.48\textwidth}
        \centering
        \begin{align*}
        &\max_{y, w \in \Real^m, \alpha \in \Real} \quad \ z^*1^\top w - \alpha\\
        \text{s.t. }\quad & \alpha + e_i^\top B y - e_i^\top A w \geq 0 \quad \forall i \in [n]  \\
             & 1^\top y = 1, \quad y \geq 0 \quad w \geq 0.  
        \end{align*}
        \caption{Exploiter's LP}
        \label{equ: exploiter-LP}
    \end{subfigure}
    \caption{\sol LPs}
    \label{fig: LPs}
\end{figure}

\paragraph{Victim's LP.} The victim must compute a security strategy as suggested in \cref{assump: victim}. From the victim's perspective, it faces a zero-sum game with payoff matrix $A$. As such, a security strategy can be computed using the standard LP formulation of computing an NE for this zero-sum game~\citep{NashLP}. This LP can be seen in \cref{equ: victim-LP}.

\begin{proposition}\label{prop: victim-strategy}
    For any optimal solution $(x^*,z^*)$ to LP~\ref{equ: victim-LP}, $x^*$ is a \sol strategy for the victim and $z^* = p_v$ is the victim's guaranteed payoff. Furthermore, the set of all \sol strategies for the victim, $X^*$, is the non-empty polytope:
    \begin{equation*}
        X^* = \set{x \in \Delta(n) \mid z^* \leq x^\top A e_j \ \forall j \in [m]}.
    \end{equation*}
\end{proposition}


\paragraph{Exploiter's LP.} The exploiter computes its worst-case best-response to $X^*$ as suggested in~\cref{assump: exploiter}.  
The inner minimization problem in~\cref{equ: exploiter} amounts to first fixing an exploiter strategy $y \in \Delta(m)$ and then identifying the worst-case victim strategy for the exploiter, $\min_{x \in X^*} x^{\top} B y$. 
Notice the optimization problem involves the implicitly defined set $X^*$.
Our key technical insight is that the exploiter can avoid enumerating $X^*$ by dualizing the problem. Observe that the inner minimization problem can be explicitly written as an LP:
\begin{equation}\label{equ: inter-LP}
\begin{aligned}
\min_{x \geq 0} \quad & \ x^\top B y \\
\text{s.t. } \quad & z^* - x^\top A e_j \leq 0 \quad \forall j \in [m] \\
	     & 1^\top x - 1 = 0,
\end{aligned}
\end{equation}
where $(x^*,z^*)$ is any solution to LP~\ref{equ: victim-LP}, which can be constructed by the exploiter as it knows the victim's payoff matrix $A$. The exploiter should optimize its choice of $y$ by maximizing the value of \cref{equ: inter-LP}, which would result in a max-min problem.
We can convert this problem into a double-max problem by taking the dual of \cref{equ: inter-LP}, given as follows:
\begin{equation}\label{equ: Dual}
    \begin{aligned}
        \max_{w\geq 0, \alpha} \quad & \ z^*1^\top w - \alpha\\
    	\text{s.t. } \quad &
    	\alpha + e_i^\top By - e_i^\top A w \geq 0 \quad \forall i \in [n],
    \end{aligned}
\end{equation}
where $w \in \Real^{m}_{\geq 0}$ and $\alpha \in \Real$ are the dual variables corresponding to the first and second constraints in \cref{equ: inter-LP}, respectively. 
Applying $\max_{y \in \Delta(m)}$ on top of~\cref{equ: Dual}, yields the exploiter's LP, which can be seen in \cref{equ: exploiter-LP}. We give the full derivation in the Appendix.

\begin{proposition}\label{prop: exploiter-strategy}
    For any optimal solution $(y^*,w^*, \alpha^*)$ to LP~\ref{equ: exploiter-LP}, $y^*$ is a \sol strategy for the exploiter and $z^*1^\top w^* - \alpha^* = p_e$ is the exploiter's guaranteed payoff. Furthermore, the set of all \sol strategies for the exploiter, $\BR(X^*)$, is the non-empty projection polytope:
    \begin{equation*}
        \BR(X^*) = \set{y \in \Delta(m) \mid \exists w \in \Real^m_{\geq 0}, \ e_i^\top B y + (z^* 1^\top - e_i^\top A) w \geq p_e \  \forall i \in [n]}.
    \end{equation*}
    In particular, $\BR(X^*)$ contains the non-empty polytope induced by fixing $w = w^*$.
\end{proposition}
Note that the exploiter must first solve LP~\ref{equ: victim-LP} to compute $z^*$. Effectively, the exploiter is using its information advantage to simulate what the victim will do. After computing $z^*$, the exploiter can then compute its strategy $y^*$ using LP~\ref{equ: exploiter-LP}.
\begin{theorem}\label{prop: sol-efficient}
    LP~\ref{equ: victim-LP} has $n+1$ variables and $m+1$ constraints, and LP~\ref{equ: exploiter-LP} has $2m+1$ variables and $n+1$ constraints. Thus, both LPs have polynomially many variables and constraints, so both players can compute their \sol strategy for a general-sum game independently in polynomial time.
\end{theorem}


\paragraph{Approximate Rationality.} Although we focus on worst-case rationality, our framework still applies when the victim is willing to take calculated risks. Suppose the victim will accept any $\epsilon$-maximin strategy for some $\epsilon \geq 0$, i.e. any strategy $x^*$ satisfying $\min_{y} {x^*}^\top A y \geq \max_{x} \min_y  {x}^\top A y - \epsilon$. Clearly, the victim's approximate maximin set is still a polytope and is given by,
\begin{equation*}
    X^*_{\epsilon} = \set{x \in \Delta(n) \mid z^* - \epsilon \leq x^\top A e_j \ \forall j \in [m]}.
\end{equation*}
We then see the best-response for the exploiter to the victim's set of strategies takes the same form; the only difference is the objective function of LP~\ref{equ: exploiter-LP} is changed to be $(z^*-\epsilon)1^\top w - \alpha$. 

 


\section{Algorithms for Computing \sol in Markov Games}\label{sec: markov}

We can extend the techniques above to compute \sol solutions for Markov games. A naive approach is to simply use the techniques from \cref{sec: bimatrix} directly on the bimatrix representation of the Markov game. However, this approach is computationally inefficient since the bimatrix representation is exponentially sized in general. 
In fact, general \sol solutions to a Markov game, similar to general NE, consist of history-dependent policies --- simply writing down these policies requires space exponential in the horizon $H$. 
Consequently, as traditionally done in the Markov game literature, we focus on computing Markovian policies for both players, which can be implemented in practice. To this end, we require a slightly stronger assumption on the players' rationality.

\begin{assumption}\label{assump: markovian}
    Both players are worst-case rational in every stage game and this is common knowledge.
\end{assumption}
In words, the players not only wish to optimize their overall worst-case payoff, but they also wish to do well against the stage game they may face at any moment in time.
This leads the players to computing even more structured solutions: namely, Markov-perfect \sol solutions. 

\paragraph{\mpsol Solutions.} The \sol solutions for the stage game $(h,s)$ are defined as,
\begin{equation*}
    \Pi_{hs}^* := \argmax_{\pi \in \Pi} \min_{\nu \in N} V^{\pi,\nu}_{hv}(s) \text{  and  } N_{hs}^* := \argmax_{\nu \in N} \min_{\pi^* \in \Pi_{hs}^* \cap \bigcap_{k > h, s'} \Pi^*_{ks'}} V^{\pi^*,\nu}_{he}(s).
\end{equation*}
The definition of $\Pi^*_{hs}$ is standard; however, notice the argument of the minimum in $N^*_{hs}$ reflects the exploiter's knowledge that the victim must be optimal at every future stage, not just at stage $(h,s)$.
We define a policy pair $(\pi^*, \nu^*)$ to be a \emph{Markov-perfect \sol} (\mpsol) solution if $(\pi^*, \nu^*)$ is a \sol solution for every stage game, i.e. $\pi^* \in \Pi^*_{h,s}$ and $\nu^* \in N^*_{h,s}$ for all $h,s$. 

\begin{definition}(\mpsol)
    \begin{equation*}
        \mpsol := \bigcap_{h,s} \Pi_{hs}^* \times \bigcap_{h,s} N^*_{hs}.
    \end{equation*}
\end{definition}
Similar to Markov-perfect NE, a \mpsol solution always exists, is a pair of Markovian policies, and is a \sol solution, so inherits all the properties we saw from \cref{prop: sol}. 

\begin{proposition}\label{prop: mp=viser}
    Under \cref{assump: markovian}, $\varnothing \neq \mpsol \subseteq \sol$ and so players will play some $(\pi^*,\nu^*) \in \mpsol$.
\end{proposition}




\begin{figure}
\begin{multicols}{2}
\begin{center}
\begin{algorithm}[H]
\caption{\textsc{Victim}}\label{alg: victim}
\SetAlgoLined
\KwIn{Markov Game $G$ without $\{R_{he}\}_h$}
 \For{$h = H$ down to $1$}{
    \For{$s \in S$}{
        $Q^*_{hv}(s) \gets$ \cref{equ: Q*} \\
        $\pi_h^*(s), V_{hv}^*(s) \gets$ LP~\ref{equ: victim-LP}($Q^*_{hv}(s)$)
    }
 }
\ \\
\ \\
\Return $\pi^* = \{\pi_h^*(s)\}_{h,s}$
\end{algorithm}
\end{center}
\hfill
\begin{center}
\begin{algorithm}[H]
\caption{\textsc{Exploiter}}\label{alg: exploiter}
\SetAlgoLined
\KwIn{Markov Game $G$}
 \For{$h = H$ down to $1$}{
    \For{$s \in S$}{
        $Q^*_{hv}(s), Q^*_{he}(s) \gets$ \cref{equ: Q*} \\
        $\pi_h^*(s), V_{hv}^*(s) \gets$ LP~\ref{equ: victim-LP}($Q^*_{hv}(s)$) \\
        $\nu_h^*(s), w_h^*(s), \alpha_h^*(s)$ $\gets$ \\ LP~\ref{equ: exploiter-LP}($Q^*_{hv}(s)$, $Q^*_{he}(s)$)
    }
 }
 \Return $\nu^* = \{\nu_h^*(s)\}_{h,s}$
\end{algorithm}
\end{center}
\end{multicols}
\caption{\sol Algorithms}
\end{figure}

\paragraph{Algorithms.} In the same spirit as Nash-VI~\citep{Markov, survey, NashVI}, both players can use backward induction: each can solve their respective LP for the bimatrix stage games backwards in time. 
For any $s \in S$, the stage game at the final time step is just the last bimatrix game: $Q^*_{Hu}(s) = R_{hu}(s)$.
For each $h < H,s \in S,u \in \{v,e\}$, the bimatrix stage game, $Q^*_{hu}(s)$, is defined by,
\begin{equation}\label{equ: Q*}
    Q^*_{hu}(s)[a_v,a_e] = \brac{R_{hu}(s,a_v,a_e) + \sum_{s'} P_h(s' \mid s, a_v, a_e) V^*_{(h+1)u}(s')},
\end{equation}
where $V_{(h+1)v}^*(s')$ is the maximin value induced by $\Pi^*_{(h+1)s'}$, and 
and $V_{(h+1)e}^*(s')$ is the constrained maximin value induced by $N^*_{(h+1)s'}$.
Overall, for $h$ in backward order and all $s$, each player computes $Q^*_{hu}(s)$ using \cref{equ: Q*} and then solves their bimatrix LP with input matrices $A = Q^*_{hv}(s)$ and $B = Q^*_{he}(s)$. The product policy constructed from all stages yields a \mpsol for each player.

\paragraph{Victim's Algorithm.} The victim needs to compute a maximin for each stage game. 
For each $h,s$ in backward order, the victim computes $Q^*_{hv}(s)$ using \cref{equ: Q*} and the previously computed future stage values, $V^*_{(h+1)v}(s')$. Then, it solves LP~\ref{equ: victim-LP} with input matrix $A = Q^*_{hv}(s)$ to compute a strategy $\pi^*_h(s)$ and the current stage value $V^*_{hv}(s)$. The total product policy $\pi^* = \{\pi^*_h(s)\}_{h,s}$ is a \mpsol for the victim. The victim's full algorithm can be seen in \cref{alg: victim}. The correctness of \cref{alg: victim} relies critically on the fact that the Nash value of each zero-sum stage game, and hence the maximin value of each stage game, is unique. 

\begin{proposition}\label{prop: victim-mstrategy}
    For any output $\pi^*$ of \cref{alg: victim}, $\pi^*$ is a \mpsol strategy for the victim, and $\sum_{s} \mu(s) V_{1v}^*(s) = p_v$ is the victim's guaranteed payoff for the game. Furthermore, the set of all \sol strategies for the victim in the stage game $(h,s)$ is the non-empty polytope,
    \begin{equation*}
    \hat \Pi_{hs} := \set{\pi_h(s) \in \Delta(n) \mid V_{hv}^*(s) \leq \pi_h(s)^\top Q^*_{hv}(s) e_j \ \forall j \in [m]},
    \end{equation*}
    and the set of \mpsol policies for the game is the non-empty polytope: $\bigcap_{h,s} \Pi^*_{hs} = \bigtimes_{h,s} \hat \Pi_{hs}$. 
\end{proposition}

\paragraph{Exploiter's Algorithm.} The exploiter needs to compute a worst-case best-response for each stage game. First, the exploiter can run \cref{alg: victim} to compute the optimal values $V^*_{hv}(s)$ that it can use to compute the matrix $Q^*_{hv}(s)$ through \cref{equ: Q*}.  Note, the exploiter can accurately compute the victim's $Q^*_{hv}(s)$ since the maximin value of each future-time stage game, $V_{(h+1)v}^*(s')$, is unique. 
Next, for each $h,s$ in backward order, it computes $Q^*_{he}(s)$ using \cref{equ: Q*} and the already computed stage values $V^*_{(h+1)s'}$. Then, using $V_{hv}^*(s)$, it solves LP~\ref{equ: exploiter-LP} for the bimatrix game $(Q_{hv}(s),Q_{he}(s))$ to compute a strategy $\nu_h^*(s)$. 
The total product policy $\nu^* = \{\nu^*_h(s)\}_{h,s}$ is a \mpsol for the victim. In fact, as the victim's algorithm is also recursive, the exploiter can simultaneously compute both $Q^*_{hv}(s)$ and $Q^*_{he}(s)$ in a single loop.
The exploiter's algorithm is summarized in~\cref{alg: exploiter}.


\begin{proposition}\label{prop: exploiter-mstrategy}
    For any output $\nu^*$ of \cref{alg: victim}, $\nu^*$ is a \mpsol strategy for the exploiter, and $\sum_{s} \mu(s) V_{1e}^*(s) = p_e$ is the exploiter's guaranteed payoff for the game. Furthermore, the set of all \sol strategies for the exploiter in the stage game $(h,s)$ is the non-empty projection polytope:
    \begin{equation*}
    \hat N_{hs} := \bigcup_{w\in \Real^m_{\geq 0}}\set{\nu_h(s) \in \Delta(m) \vert e_i^\top Q^*_{he}(s) \nu_h(s) + (V_{hv}^*(s) 1^\top - e_i^\top Q^*_{hv}(s)) w \geq p_{he}(s) \  \forall i \in [n]}.
    \end{equation*}
    and the set of \mpsol policies for the game is the non-empty polytope: $\bigcap_{h,s} N^*_{hs} = \bigtimes_{h,s} \hat N_{hs}$.
\end{proposition}

Observe the computations the players must perform in \cref{alg: victim} and \cref{alg: exploiter} can be carried out independently and using only the information available to them. We also note the algorithms' runtimes are dominated by the time to solve $O(H|S|)$ LPs. Each individual LP has polynomially many variables and constraints with respect to the size of the original Markov game, so can be solved in polynomial time. As there are polynomially many such LPs, each player's algorithm runs in polynomial time.

\begin{theorem}\label{thm: mpsol-efficient}
    Both players can independently compute a \mpsol policy for a general-sum Markov game in polynomial time. 
\end{theorem}




\section{Experiments}\label{sec: experiments}

We illustrate \sol on a generalization of the game from \cref{table: exploiter-game} and uniformly random games\footnote{Code is available at \url{https://github.com/jermcmahan/VISER}.}. Let $A$ and $B$ denote the payoff matrices from \cref{table: exploiter-game} with choice of $x = 10$ and $y = 10$. We consider the block-diagonal generalizations (with $r$ blocks) defined by,
\begin{equation*}
    \hat A = \begin{pmatrix}
        A & 0 & \ldots & 0 \\
        0 & A & \ldots & 0 \\
        \vdots & \vdots & \ldots & \vdots \\
        0 & 0 & \ldots & A 
    \end{pmatrix}
    \text{ and } 
    \hat B = \begin{pmatrix}
        B & 0 & \ldots & 0 \\
        0 & B & \ldots & 0 \\
        \vdots & \vdots & \ldots & \vdots \\
        0 & 0 & \ldots & B 
    \end{pmatrix}
\end{equation*}
It is easy to verify the form of \sol solution for $(\hat A, \hat B)$. The victim chooses either M or U in each block and evenly mixes between them. Formally, the victim's strategy $x$ must satisfy $x_{3i} = 0$, $x_{3i-2} + x_{3i-1} = 1/r$, $x_{3i-2} \not = x_{3i-1}$, and $x_i \in \{0, 1/r\}$. The payoff the victim receives is then $p_v = 10/r$. The exploiter's best response to this set of strategies is any odd column, which achieves it payoff of at least $p_e = 10/r$.

We then extend these bimatrix games into Markov games. We define $S = [10]$, $A = [n] \times [m]$, and $H = 10$. The reward matrices are exactly $\hat A$ and $\hat B$ at each step: $R_{hv}(s) = \hat A$ and $R_{he}(s) = \hat B$ for all $h, s$. For simplicity, we equip the game with uniform transitions $P(s' \mid s, a, b) = \frac{1}{|S|}$. Lastly, we assume the initial state is state $1$. Given the uniform transition and identical per-step payoff matrices, the \mpsol solution just consists of a \sol to each per-step game (each stage game is just the original game plus $(H-h)$ times a constant and so solutions are not affected). Thus, the expected values are at least $p_v = 10H/r$ and $p_e = 10H/r$ respectively.

\begin{figure}
    \begin{subfigure}[b]{.32\textwidth}
        \centering
        \includegraphics[scale=.3]{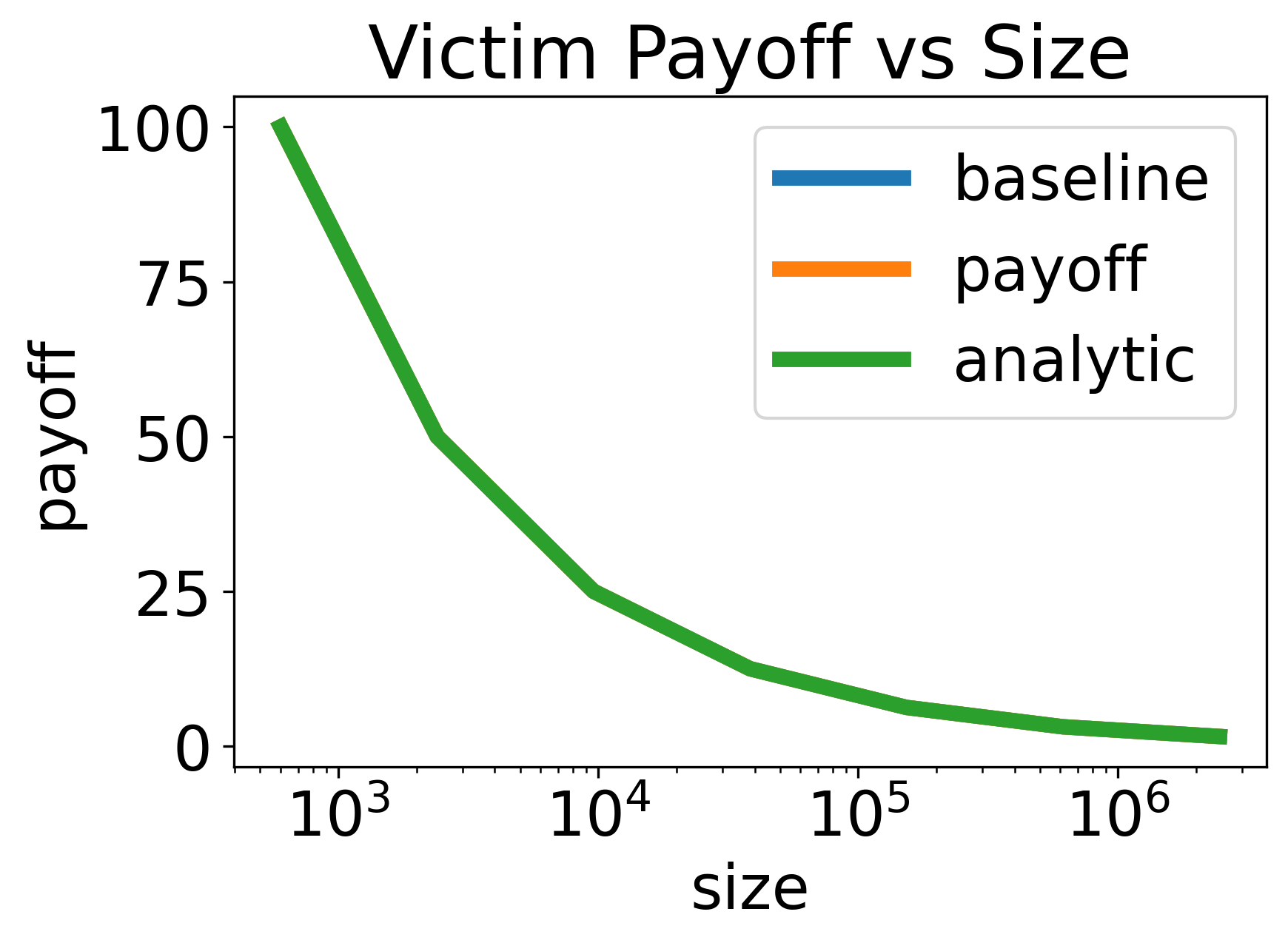}
        \caption{Victim}
        \label{fig: victim1}
    \end{subfigure}
    \hfill
    \begin{subfigure}[b]{.32\textwidth}
        \centering
        \includegraphics[scale=.3]{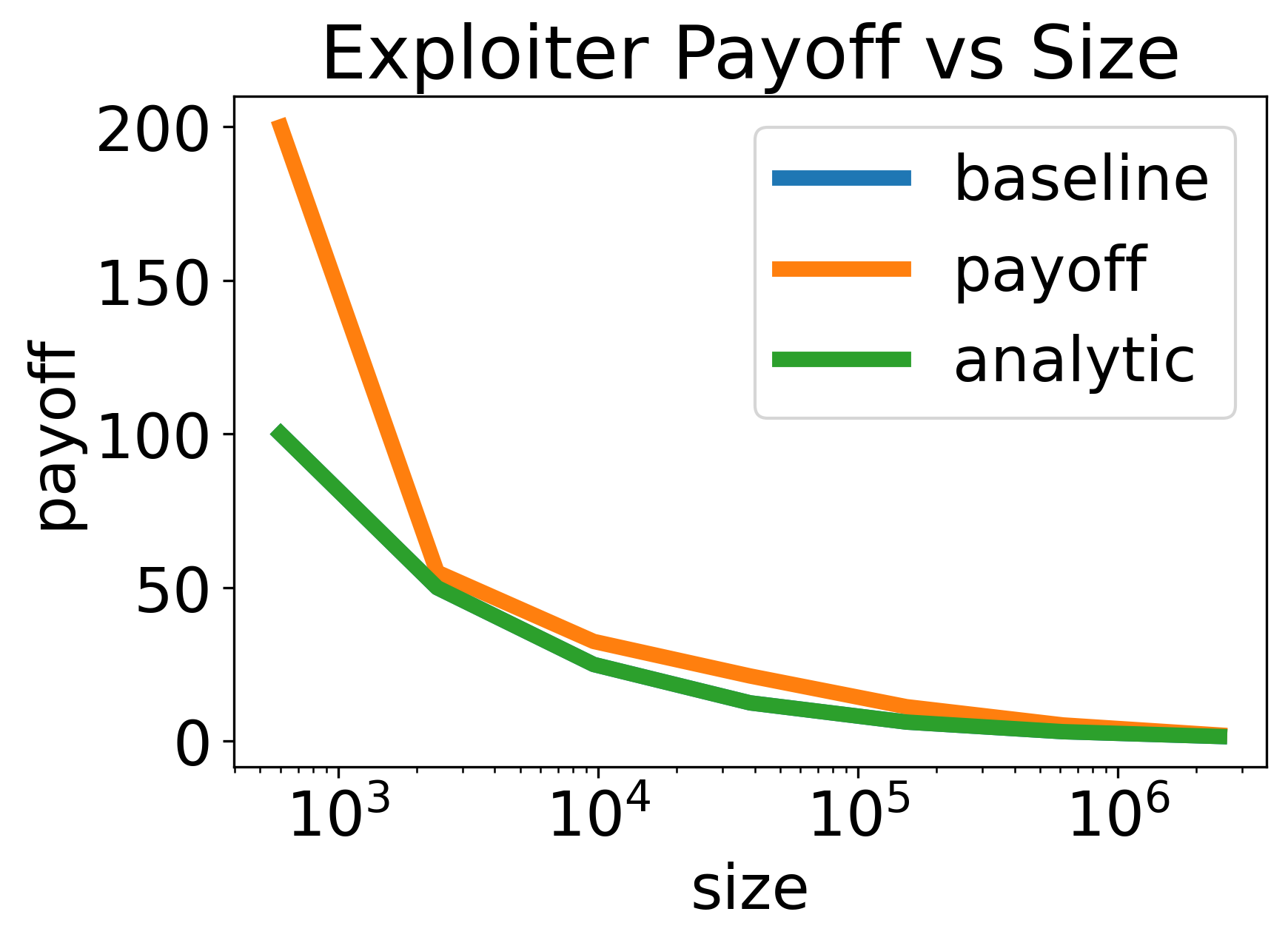}
        \caption{Exploiter}
        \label{fig: exploiter1}
    \end{subfigure}
    \hfill
    \begin{subfigure}[b]{.32\textwidth}
        \centering
        \includegraphics[scale=.3]{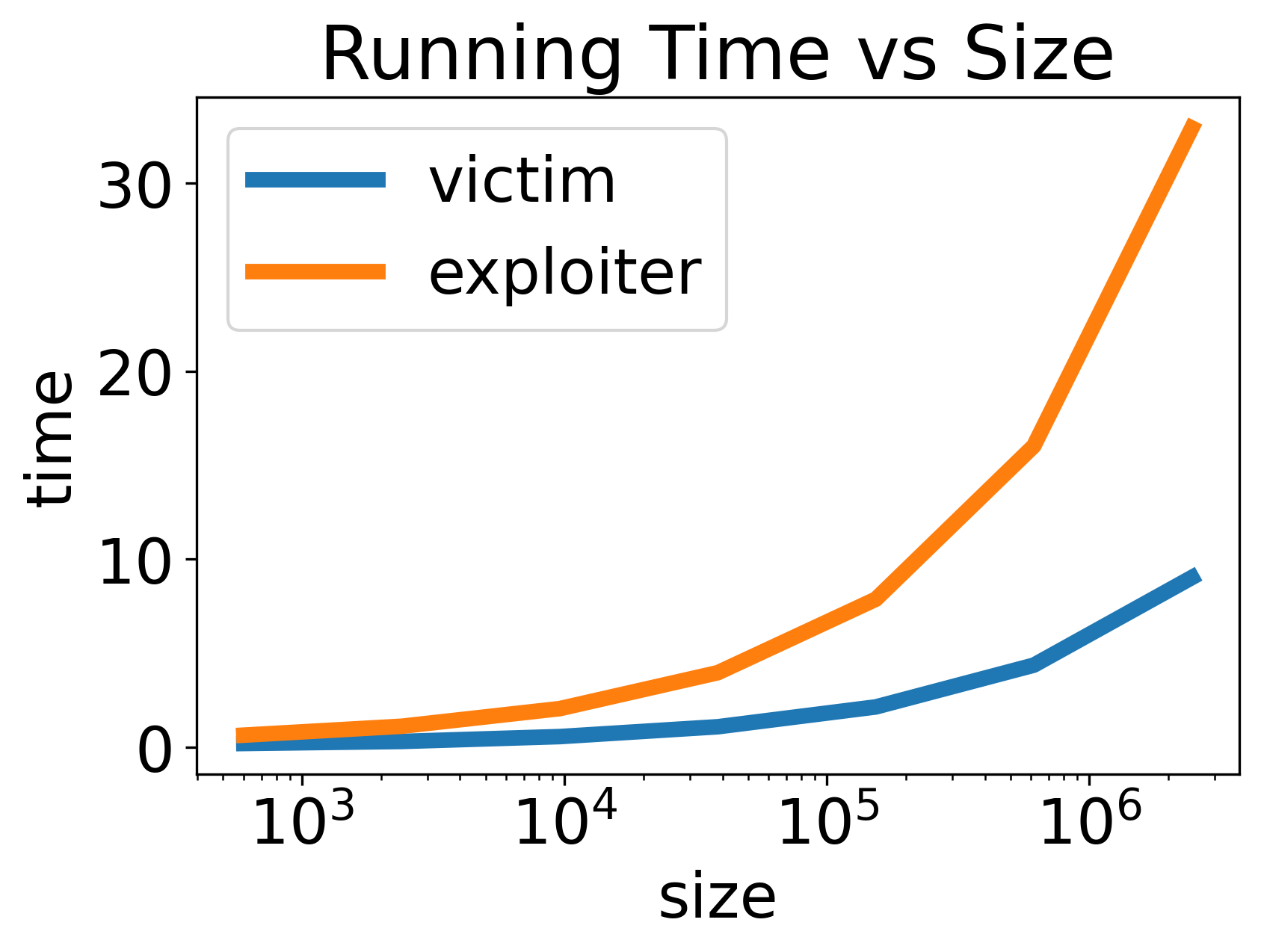}
        \caption{Time}
        \label{fig: time1}
    \end{subfigure}
    \caption{Block Diagonal Example}
    \label{fig: block}
\end{figure}

Running \cref{alg: victim} on these Markov games with $r$ varying (and thus the number of entries varying) yields the plot in \cref{fig: victim1}. Here, the baseline curve corresponds to plotting the computed $p_v$, the payoff curve corresponds to plotting $V_v^{\pi^*,\nu^*}$ (the actual value achieved), and the analytic curve corresponds to plotting the exact formula we derived for $p_v = 10H/r$. We see that all three match nearly exactly for this game. We also run \cref{alg: exploiter} on the game to create \cref{fig: exploiter1}. Here, we see the actual payoff achieved by the exploiter is better than the worst-case baseline of $p_e$. In all cases, we see the analytic solution matches the computed baseline and so the algorithms worked correctly. In \cref{fig: time1}, we see the algorithms were able to compute \mpsol solutions for Markov games with far over a million entries (measured by $HSnm$) in under a minute (note that time is measured in seconds). Unsurprisingly, the victim's algorithm ran significantly faster than the exploiter's.

\begin{figure}
    \begin{subfigure}[b]{.32\textwidth}
        \centering
        \includegraphics[scale=.3]{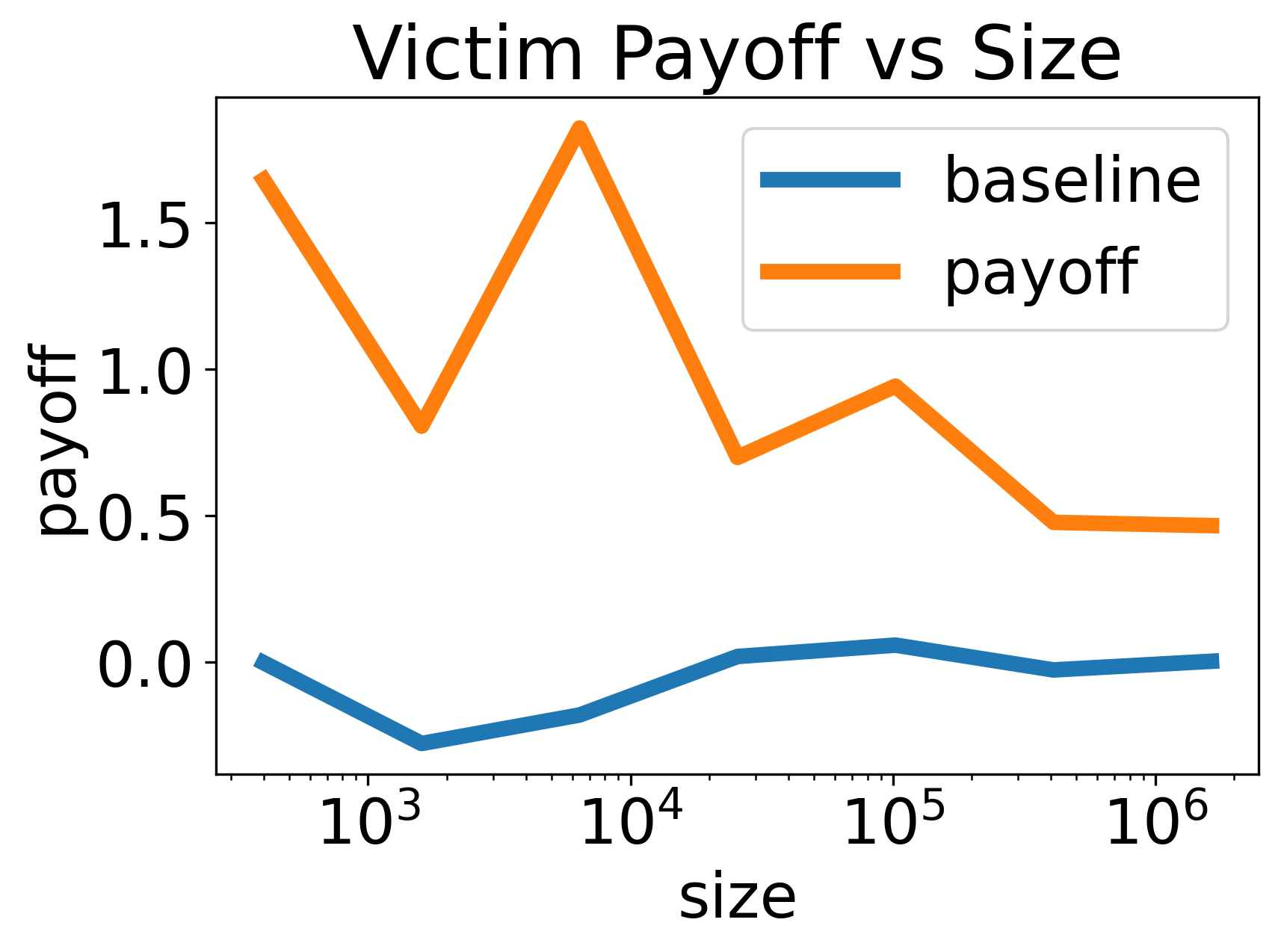}
        \caption{Victim}
        \label{fig: victim2}
    \end{subfigure}
    \hfill
    \begin{subfigure}[b]{.32\textwidth}
        \centering
        \includegraphics[scale=.3]{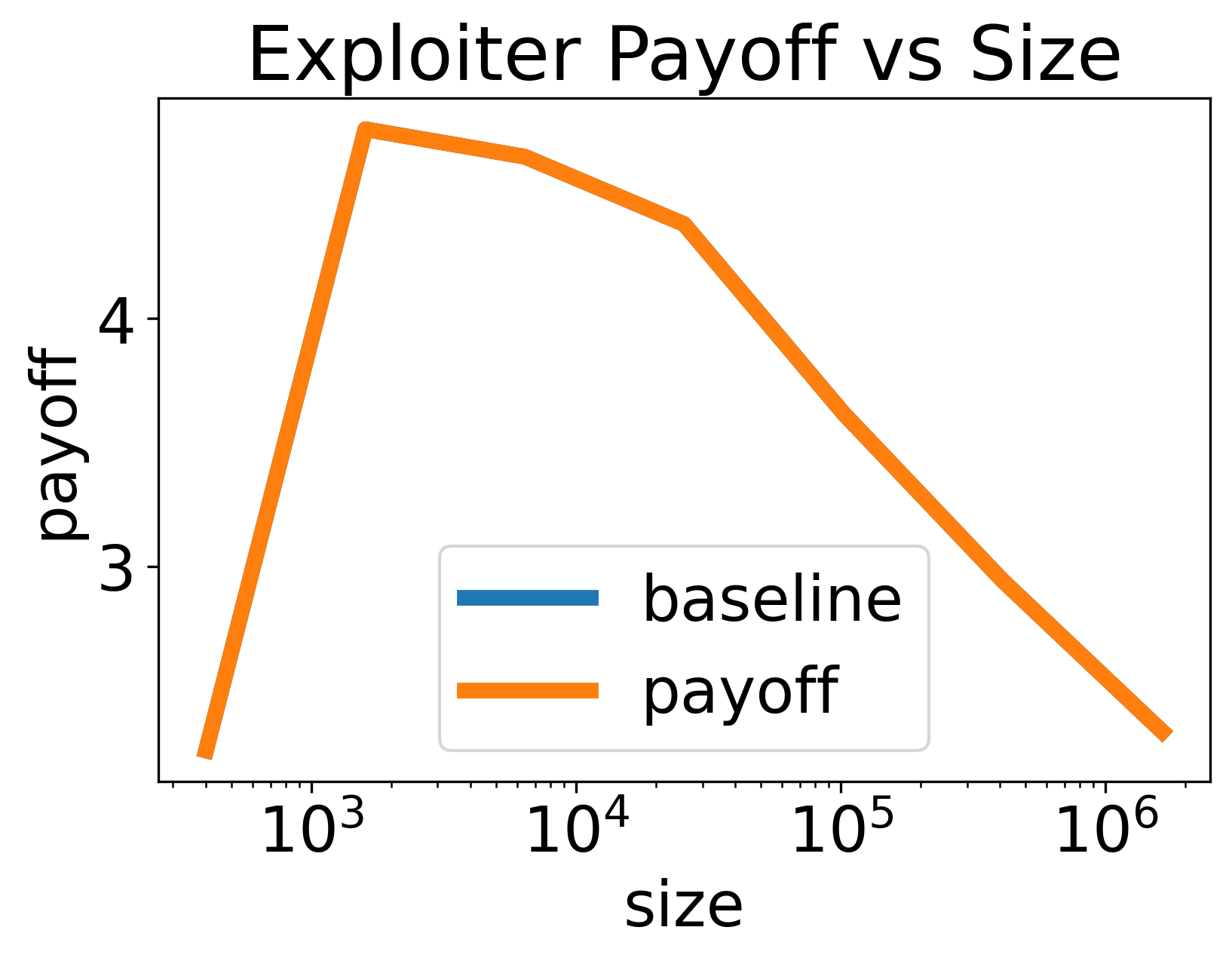}
        \caption{Exploiter}
        \label{fig: exploiter2}
    \end{subfigure}
    \hfill
    \begin{subfigure}[b]{.32\textwidth}
        \centering
        \includegraphics[scale=.3]{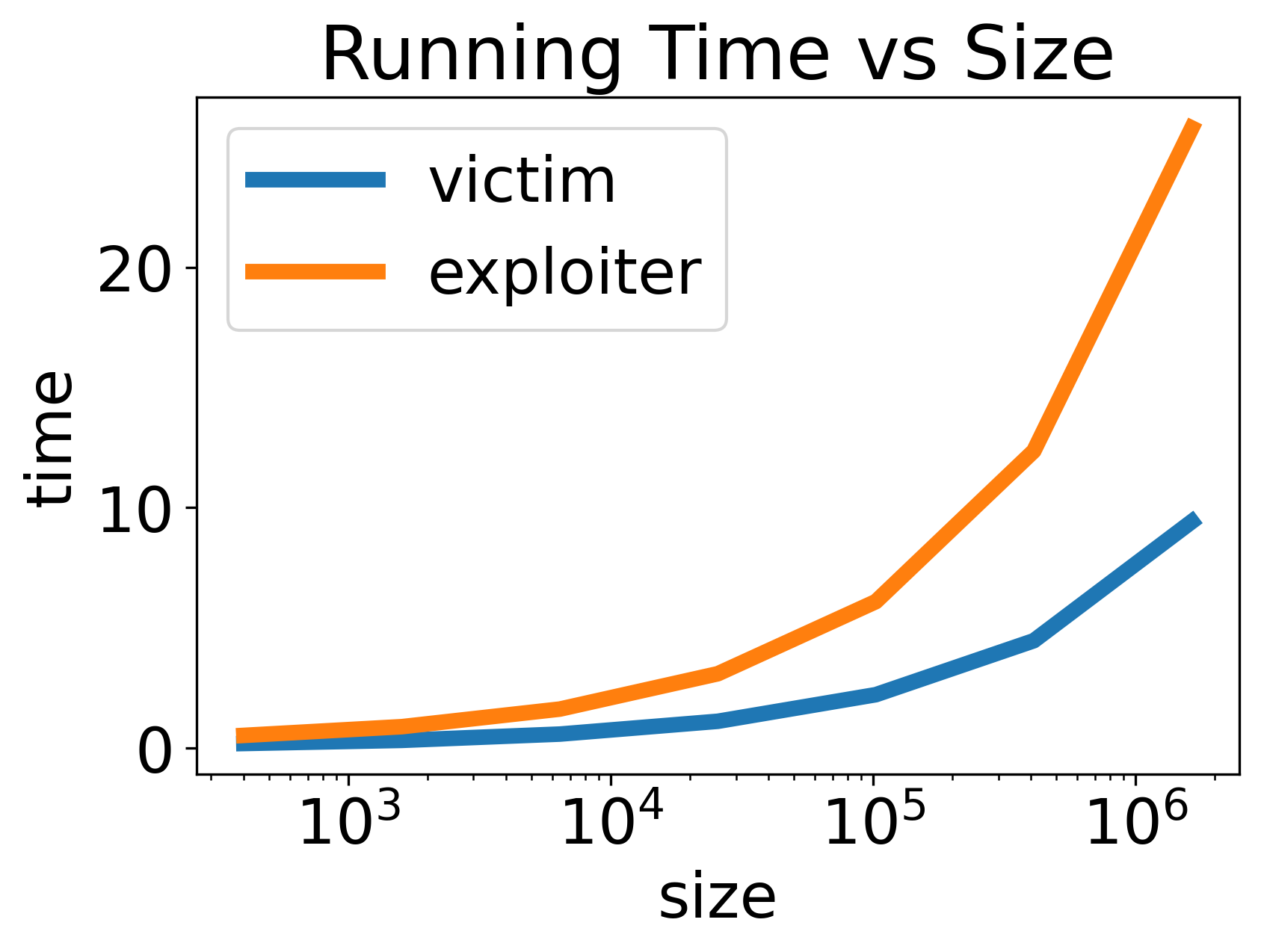}
        \caption{Time}
        \label{fig: time2}
    \end{subfigure}
    \caption{Random Game Example}
    \label{fig: random}
\end{figure}

We also tested the algorithms on Markov games with uniformly random transitions (again with $|S| = H = 10$), but with a uniformly-random bimatrix game at each step. Formally, each $R_{hu}(s) \in U[-1,1]^{n \times n}$ with $n$ varying per Markov game. In such games, a \mpsol solution no longer is simply a \sol for each bimatrix game separately. In \cref{fig: victim2}, we see that the victim tended to achieve a value that was much higher than its baseline, maximin, payoff. However, in \cref{fig: exploiter2}, we see the exploiter was unable to do much better than its baseline payoff in such random games. Regardless, we again see our algorithms were able to handle games with well over a million entries in under a minute. This gives evidence of their efficiency in practice.

\section{Conclusions}

In this paper, we studied games with an information asymmetry structure particularly prevalent in security applications. Namely, games where a victim player only knows its own payoff but an exploiter player knows both the victim's payoffs and its own. This includes any security game where the victim would not know the true goal of the exploiter and adversarial MARL when the exploiter hacks the actions of one of the players. We introduced a new solution concept called \sol that is appropriate so long as the players are uncertainty averse and so choose to maximize their worst-case payoffs. Importantly, both players have no incentive to deviate from a \sol solution. We then showed \sol solutions for bimatrix games can be efficiently computed in a distributed fashion using LPs. We also defined \mpsol as the appropriate solution concept for Markov games and showed \mpsol solutions are \sol solutions and can also be computed efficiently in a distributed manner via backward induction.

\paragraph{Future Work.} Although \sol addresses many concerns for games with information asymmetry, there are still two notable directions left open. First, \sol only addresses two player games. When many players are involved the information structures become much more complicated. 
Beyond multiple players, our paper only concerns the ``planning'' setting where each player's own payoffs are known to them in advance. One could similarly consider a learning setting where the players do not know even their own payoffs but can learn about them by repeated play. This direction may have connections to the theory of partially observable Markov games.

\bibliographystyle{plain}
\bibliography{viser}

\newpage

\appendix

\section{Proofs}

This section contains all proofs for the theoretical results in the main paper.

\subsection{Proofs for \texorpdfstring{\cref{sec: sol}}{Section 3}}

Proof of \cref{assump: victim}.
\begin{proof}
    Suppose the victim chooses any $x \not \in X^*$. By definition of $X^*$, there must exist an exploiter strategy $y$ satisfying $x^\top A y < p_v$. Absent any additional information, the victim cannot assume the exploiter will not play $y$ and so risks achieving a lower payoff than $p_v$. Even knowing \cref{assump: exploiter} does not alleviate this risk due to the victim not knowing $B$ which would be necessary to predict $y^*$. Thus, under \cref{assump: rationality}, the victim must play some $x^* \in X^*$.
\end{proof}

Proof of \cref{assump: exploiter}.
\begin{proof}
    Suppose the exploiter chooses any $y \not \in \BR(X^*)$. By definition of $\BR(X^*)$, there must exist a victim strategy $x^* \in X^*$ satisfying $x^{*^\top} B y < p_e$. Under \cref{assump: rationality}, the exploiter knows the victim plays some strategy in $X^*$, but is uncertain which security strategy is chosen. Thus, absent any additional information, the exploiter cannot assume the victim will not play $x^*$ and so risks achieving lower payoff than $p_e$. Thus, under \cref{assump: rationality}, the victim must play some $y^* \in \BR(X^*)$.
\end{proof}

For the claim that $\sol \neq \varnothing$, 
this follows from the fact that both $X^*$ and $\BR(X^*)$ are non-empty sets as shown in \cref{prop: victim-strategy} and \cref{prop: exploiter-strategy}. The proof of \cref{prop: sol} is then immediate.

\subsection{Proofs for \texorpdfstring{\cref{sec: bimatrix}}{Section 4}}
As mentioned in the main text, the proof of \cref{prop: victim-strategy} is immediate from standard bimatrix game theory~\citep{NashLP}.

Proof of \cref{prop: exploiter-strategy}.
\begin{proof}
The exploiter wishes to compute
\[\argmax_{y \in \Delta(m)} \min_{x^* \in X^*} x^{*^\top} B y\]
For any fixed $y$, we see that the minimization problem $\min_{x^* \in X^*} x^{*^\top} B y$ can be explicitly written as an LP:
\begin{equation*}
\begin{aligned}
\min_{x \geq 0} \quad & \ x^\top B y \\
\text{s.t. } \quad & z^* - x^\top A e_j \leq 0 \quad \forall j \in [m] \\
	     & 1^\top x - 1 = 0,
\end{aligned}
\end{equation*}
where $(x^*,z^*)$ is any solution to LP~\ref{equ: victim-LP}. Next, we construct the dual of this program. To this end, we introduce a dual vector $w \in \Real^m_{\geq 0}$ corresponding to the inequality constraints and a dual variable $v \in \Real$ corresponding to the equality constraint. We multiply these dual variables by their respective constraints and add them to the objective to get the equivalent optimization:
\begin{equation*}
    \max_{w \geq 0,v} \min_{x \geq 0} x^\top B \nu + (z^* 1^\top - x^\top A)w + (x^\top 1 - 1)v
\end{equation*}
By rearranging the objective to be in terms of $x$, we get:
\begin{equation*}
    \max_{w \geq 0,v} \min_{x \geq 0} x^\top(By -Aw +1v) + z^*1^\top w - v
\end{equation*}
Moving the terms involving $x$ into the constraints then gives the Dual:
    \begin{align*}
        \max_{w\geq 0, \alpha} \quad & \ z^*1^\top w - \alpha\\
    	\text{s.t. } \quad &
    	\alpha + e_i^\top By - e_i^\top A w \geq 0 \quad \forall i \in [n],
    \end{align*}
Applying $\max_{y \in \Delta(m)}$ outside of the Dual, yields the exploiter's LP~\ref{equ: exploiter-LP}:
\begin{align*}
        &\max_{y, w \in \Real^m, \alpha \in \Real} \quad \ z^*1^\top w - \alpha\\
        \text{s.t. }\quad & \alpha + e_i^\top B y - e_i^\top A w \geq 0 \quad \forall i \in [n]  \\
        & 1^\top y = 1, \quad y \geq 0 \quad w \geq 0.  
\end{align*}

The fact that there exist optimal solutions, i.e., $\BR(X^*) \neq \varnothing$, follows from LP~\ref{equ: exploiter-LP} being feasible and bounded. Specifically, it is easily seen that choosing $y = e_1$, $w = 0$, and $\alpha = \max_{i \in [n]} |e_i^\top B e_1$ gives a feasible solution to LP~\ref{equ: exploiter-LP}. Boundedness follows from the fact that by LP duality, LP~\ref{equ: exploiter-LP} is value equivalent to the original problem $\max_{y \in \Delta(m)} \min_{x \in X^*} x^\top B y$, which is clearly bounded being that $(A,B)$ is a finite bimatrix game. 

Suppose that $y^*, w^*, \alpha^*$ are any solution to LP~\ref{equ: exploiter-LP}. Then, the optimal objective is exactly $p_e := z^*1^\top w^* - \alpha^*$. Any optimal solution triple must satisfy $z^*1^\top w - \alpha = p_e$. Adding in this constraint yields the polytope of optimal solutions:
\begin{equation*}
    P^* := \set{y \in \Delta(m), w \in \Real^m_{\geq 0} \mid e_i^\top B y + (z^* 1^\top - e_i^\top A) w \geq p_e \  \forall i \in [n]}.
\end{equation*}
To write this as a subset of the exploiter's strategy space, we simply project it: for any fixed $w$, the set of $y$s that satisfy the inequality are all optimal strategies for the exploiter. By taking the union of these sets over all $w$, we get the set of all the exploiter's optimal strategies:
\begin{equation*}
    \bigcup_{w \in \Real^m_{\geq 0}} \set{y \in \Delta(m) \mid e_i^\top B y + (z^* 1^\top - e_i^\top A) w \geq p_e \  \forall i \in [n]}.
\end{equation*}
Equivalently, we can add an existential quantifier over $w$ to get the set:
\begin{equation*}
    \BR(X^*) = \set{y \in \Delta(m) \mid \exists w \in \Real_{\geq 0}^m, \ e_i^\top B y + (z^* 1^\top - e_i^\top A) w \geq p_e \  \forall i \in [n]}.
\end{equation*}
Since $\BR(X^*)$ is a projection of the polytope $P^*$, Corollary 2.4 in ~\citep{bertsimas1997introduction} implies that $\BR(X^*)$ is a polytope.
\end{proof}

The proof of \cref{prop: sol-efficient} is immediate from \cref{prop: victim-strategy} and \cref{prop: exploiter-strategy}.

\subsection{Proofs for \texorpdfstring{\cref{sec: markov}}{Section 5}}

For a Markov game $G$ under \cref{assump: markovian}, $\sol = \Pi^* \times N^*$ where,

\begin{equation*}
     \Pi^* := \argmax_{\pi \in \Pi} \min_{\nu \in N} V_{v}^{\pi,\nu} \text{  and  } N^* := \argmax_{\nu \in N} \min_{\pi \in \Pi^*} V_e^{\pi, \nu}.
\end{equation*}
Note, absent the assumption, the max's and min's would be with respect to full history dependent policies rather than only Markovian ones.
We see \sol for each stage game $(h,s)$ is given by,
\begin{equation*}
    \Pi_{hs}^* := \argmax_{\pi \in \Pi} \min_{\nu \in N} V_{hv}^{\pi,\nu}(s) \text{  and  } N_{hs}^* := \argmax_{\nu \in N} \min_{\pi \in \Pi_{hs}^* \cap \bigcap_{k > h,s'} \Pi^*_{ks'}} V_{he}^{\pi, \nu}(s).
\end{equation*}
The corresponding maximum stage game values are,
\begin{equation*}
    V_{hv}^*(s) := \max_{\pi \in \Pi} \min_{\nu \in N} V_{hv}^{\pi,\nu}(s) \text{  and  } V_{he}^*(s) := \max_{\nu \in N} \min_{\pi \in \Pi_{hs}^* \cap \bigcap_{k > h,s'} \Pi^*_{ks'}} V_{he}^{\pi, \nu}(s).
\end{equation*}
To simplify notation, we define $\tilde \Pi_{hs} = \Pi_{hs}^* \cap \bigcap_{k > h,s'} \Pi^*_{ks'}$ and $\tilde N_{hs} = N_{hs}^* \cap \bigcap_{k > h,s'} N_{ks'}^*$.
Since $V_{hu}^{\pi,\nu}(s)$ depends only the policies from time $h$ onward, we only need to consider partial policies that are only defined at time $h$ onward starting from state $s$. Specifically, we need only consider $\pi = \pi_h(s) \cdot \{\pi_k(s')\}_{k,s'} \in \Pi_{hs} = \Delta(n) \times \bigtimes_{k > h, s'} \Delta(n)$ and $\nu = \nu_h(s) \cdot \{\nu_k(s')\}_{k,s'} \in N_{hs} = \Delta(m) \times \bigtimes_{k > h, s'} \Delta(m)$ when defining \sol for stage games. Similarly, when applying the bellman-consistency equations,
\begin{equation}\label{equ: consistency}
    V_{hu}^{\pi,\nu}(s) = \E_{a_v \sim \pi_h(s), a_e \sim \nu_h(s)}\brac{R_{hu}(s,a_v,a_e) + \sum_{s'}P_h(s' \mid s, a_v, a_e)V_{(h+1)u}^{\pi,\nu}(s')},
\end{equation}
we only need to consider the partial policies from time $h+1$ onward starting from $s'$ in the term $V_{(h+1)u}^{\pi,\nu}(s')$. Note, the base case is: $V_{Hu}^{\pi,\nu}(s) = \E_{a_v \sim \pi_H(s), a_e \sim \nu_H(s)} R_{Hu}(s,a_v,a_e)$.

Recall, we define \mpsol to consist of policies that are \sol solutions for each stage game simultaneously: 
\begin{equation*}
    \mpsol = \bigcap_{h,s} \Pi^*_{hs} \times \bigcap_{h,s} N^*_{hs}.
\end{equation*}
Here, we assume the intersection is with respect to partial policies, so $\Pi \cap \Pi^*_{Hs}$ represents all full policies whose partial policy at the last step satisfies $\pi_H(s) \in \argmax_{\pi_H(s) \in \Delta(n)}\min_{\nu_H(s) \in \Delta(m)} V_{hv}^{\pi,\nu}(s)$ only; the partial policies at other times and states are unrestricted.
Observe that if $\pi^*_s \in \Pi_{1s}^*$ for all $s$, then $\pi^* = \{\pi^*_s\}_s \in \Pi^*$ and if $\nu^*_s \in N_{1s}^*$ for all $s$, then $\nu^* = \{\nu^*_s\}_s \in N^*$. Thus, $\mpsol \subseteq \sol$. We later show that $\mpsol \neq \varnothing$.

For all $s \in S$, define
\[\hat \Pi_{Hs} = \argmax_{\pi_H(s) \in \Delta(n)}\min_{\nu_H(s) \in \Delta(m)} \E_{a_v \sim \pi_H(s), a_e \sim \nu_H(s)} R_{Hv}(s,a_v,a_e),\]
and for all $h < H$,
\[\hat \Pi_{hs} = \argmax_{\pi_h(s) \in \Delta(n)}\min_{\nu_h(s) \in \Delta(m)} \E_{a_v \sim \pi_H(s), a_e \sim \nu_H(s)} \brac{R_{hv}(s,a_v,a_e) + \sum_{s'} P_h(s' \mid s, a_v, a_e)V^*_{hv}(s')}.\]
Note, through $Q^*$, \cref{equ: Q*}, we see these definitions are equivalent to the ones we gave in the main text.

\begin{lemma}\label{lem: victim}
    For any $s \in S$, $\tilde \Pi_{Hs} = \hat \Pi_{Hs}$ and for any $h < H$, 
    \[\tilde \Pi_{hs} = \hat \Pi_{hs} \times \bigtimes_{s'} \tilde \Pi_{(h+1)s'}.\]
    Furthermore, 
    \[\tilde \Pi_{hs} = \hat \Pi_{hs} \times \bigtimes_{k > h, s'} \hat \Pi_{ks'} \neq \varnothing.\]
\end{lemma}

\begin{proof}
    The lemma follows immediately from standard zero-sum Markov game theory~\citep{Markov, maximin}. A direct proof can also be seen as a special case of the exploiter's argument that we provide next.
\end{proof}
The proof of \cref{prop: victim-mstrategy} is then immediate from \cref{prop: victim-strategy} and \cref{lem: victim}.


\paragraph{Optimality Equations.} To prove the correctness of the exploiter's algorithm, we first derive bellman-optimality-like equations for the exploiter's stage value.

\begin{lemma}\label{lem: recursion}
    For any $s \in S$, the exploiter's \sol value $V_{he}^*(s)$ satisfies the following recursive relation:
    \[V_{He}^*(s) = \max_{\nu_H(s) \in \Delta(m)} \min_{\pi_H(s) \in \hat \Pi_{Hs}}  \E_{\pi_H(s), \nu_H(s)} R_{He}(s, a_v, a_e),\]
    and,
    \[V_{he}^*(s) = \max_{\nu_h(s) \in \Delta(m)} \min_{\pi_h(s) \in \hat \Pi_{hs}} \E_{\pi_h(s), \nu_h(s)}\brac{R_{he}(s,a_v,a_e) + \sum_{s'} P_h(s' \mid s,a_v,a_e) V^*_{(h+1)e}(s')},\]
    where $a_v \sim \pi_h(s)$ and $a_e \sim \nu_h(s)$ in both expectations.
\end{lemma}

\begin{proof}
    We proceed by induction on $h$. For the base case, consider the final time step $h = H$ and fix any $s \in S$. By the bellman-consistency equations~\ref{equ: consistency}, we know that for any $\pi$ and $\nu$, $V_{He}^{\pi,\nu}(s) = \E_{\pi_H(s),\nu_H(s)} R_{He}(s,a_v,a_e)$. We also know that by \cref{lem: victim} that $\tilde \Pi_{Hs} = \hat \Pi_{hs}$ and by the definition of the set of partial policies $N_{Hs} = \Delta(m)$. Putting these together shows,
    \begin{align*}
        V^*_{He}(s) &= \max_{\nu \in N_{Hs}} \min_{\pi \in \tilde \Pi_{Hs}} V^{\pi, \nu}_{He}(s) \\
        &= \max_{\nu \in N_{Hs}} \min_{\pi \in \tilde \Pi_{Hs}} \E_{\pi_H(s),\nu_H(s)} R_{He}(s,a_v,a_e) \\
        &= \max_{\nu_h(s) \in \Delta(m)} \min_{\pi_h(s) \in \hat \Pi_{Hs}} \E_{\pi_H(s),\nu_H(s)} R_{He}(s,a_v,a_e).
    \end{align*}

    For the inductive step, consider any time step $h < H$ and fix any $s \in S$. Applying the bellman-consistency equations~\ref{equ: consistency} to the definition of $V^*_{he}(s)$ yields:
    \begin{equation*}\label{equ: cons-star}
    V^*_{hs}(s) = \max_{\nu \in N_{hs}} \min_{\pi \in \tilde \Pi_{hs}}\E_{a_v \sim \pi_h(s), a_e \sim \nu_h(s)}\brac{R_{hu}(s,a_v,a_e) + \sum_{s'}P_h(s' \mid s, a_v, a_e)V_{(h+1)u}^{\pi,\nu}(s')}.
    \end{equation*}
    Observe that the expression decomposes: the expectation only considers the policies at the current state and time, $(\pi_h(s), \nu_h(s))$, and the summation only considers the policies at future time steps, $\Pi_{(h+1)s'} \times N_{(h+1)s'}$. Consequently, we can break down the $\max_{\nu \in N_{hs}}$ into the separate optimizations: $\max_{\nu_h(s) \in \Delta(m)}$ and $\max_{\nu \in N_{(h+1)s'}}$ for each $s' \in S$. Similarly, from \cref{lem: victim} we know that $\tilde \Pi_{hs} = \hat \Pi_{hs} \times \bigtimes_{s'} \tilde \Pi_{(h+1)s'}$ and so we can break down the $\min_{\pi \in \tilde \Pi_{hs}}$ into the separate optimizations: $\min_{\pi_h(s) \in \hat \Pi_{hs}}$ and $\min_{\pi \in \tilde \Pi_{(h+1)s'}}$ for each $s' \in S$. 
    This yields the equivalent optimization:
    \begin{equation*}
        \max_{\nu_h(s) \in \Delta(m)} \max_{\nu \in \bigtimes_{s'} N_{(h+1)s'}} \min_{\pi_h(s) \in \hat \Pi_{hs}} \min_{\pi \in \bigtimes_{s'}\tilde \Pi_{(h+1)s'}}\E_{\pi_h(s),\nu_h(s)}\brac{\ldots}.
    \end{equation*}
    Now, consider the summation term inside of the optimization: \begin{equation*}
        \E_{\pi_h(s), \nu_h(s)} \brac{\sum_{s'} P_h(s' \mid s, a_v, a_e) V_{(h+1)e}^{\pi, \nu}(s')}.
    \end{equation*} 
    We can apply linearity of expectation to get the equivalent term:
    \begin{equation*}
        \sum_{s'} \E_{\pi_h(s), \nu_h(s)}\brac{P_h(s' \mid s, a_v, a_e)V_{(h+1)e}^{\pi, \nu}(s')}.
    \end{equation*}
    Also, since $V_{(h+1)e}^{\pi, \nu}(s')$ depends only on the partial policies at future steps, $V_{(h+1)e}^{\pi, \nu}(s')$ is constant with respect to $(\pi_h(s),\nu_h(s))$ so can be pulled out of the summation to get the equivalent term: 
    \begin{equation*}
        \sum_{s'} V_{(h+1)e}^{\pi, \nu}(s')\E_{\pi_h(s), \nu_h(s)} \brac{P_h(s' \mid s, a_v, a_e)}.
    \end{equation*}
    Now, by the induction hypothesis, we know for any $s'$ at time $h+1$,
    \begin{align*}
        V^*_{(h+1)e}(s') &= \max_{\nu_{h+1}(s') \in N_{(h+1)s'}} \min_{\pi_{h+1}(s') \in \tilde \Pi_{(h+1)s'}} \\
        &\E_{\pi_{h+1}(s'),\nu_{h+1}(s')} \brac{R_{(h+1)e}(s',a_v,a_e) + \sum_{s'} P_{h+1}(s'' \mid s', a_v, a_e) V^*_{(h+2)e}(s'')}.
    \end{align*}
    Since the term $V_{(h+2)e}^*(s'')$ is fixed and shared amongst all $s'$ at time $h+1$, we see the only variation in the stage value $V^*_{(h+1)e}(s')$ comes from the choice of $(\pi_{h+1}(s'),\nu_{h+1}(s'))$ (i.e. varying the future partial policy cannot increase the objective value). These can be independently chosen for all $s'$ at time $h+1$. Thus, the optimization problems $\max_{\nu \in N_{(h+1)s'}}\min_{\pi \in \tilde \Pi_{(h+1)s'}} V^{\pi,\nu}_{(h+1)e}(s') = V^*_{(h+1)e}(s')$ are separable over $s'$. Thus, we can bring the maximin over partial policies into the summation to get the term: 
    \begin{equation*}
        \sum_{s'} \max_{\nu \in N_{(h+1)s'}} \min_{\pi \in \tilde \Pi_{(h+1)s'}} V_{(h+1)e}^{\pi, \nu}(s')\E_{\pi_h(s), \nu_h(s)}\brac{P_h(s' \mid s, a_v, a_e)}.
    \end{equation*}
    Since $V^*_{(h+1)e}(s') = \max_{\nu \in N_{(h+1)s'}} \min_{\pi \in \tilde \Pi_{(h+1)s'}} V_{(h+1)e}^*(s')$, the expression becomes:
    \begin{align*}
        \sum_{s'} V_{(h+1)e}^*(s')\E_{\pi_h(s), \nu_h(s)}\brac{P_h(s' \mid s, a_v, a_e)}. 
    \end{align*}
    As $V^*_{(h+1)e}(s')$ is still constant with respect to $(\pi_h(s), \nu_h(s))$, we can reverse the previous steps of pulling out this term and applying linearity of expectation to get the final expression:
    \[V^*_{he}(s) = \max_{\nu_h(s) \in \Delta(m)} \min_{\pi_h(s) \in \hat \Pi_{hs}} \E_{\pi_h(s), \nu_h(s)}\brac{R_{he}(s,a_v,a_e) + \sum_{s'} P_h(s' \mid s,a_v,a_e) V^*_{(h+1)e}(s')}.\]
\end{proof}

\paragraph{Solutions Sets.} For any $s \in S$, inductively define:
\[\hat N_{Hs} = \argmax_{\nu_H(s) \in \Delta(m)} \min_{\pi_H(s) \in \hat \Pi_{Hs}}  \E_{\pi_H(s), \nu_H(s)} R_{He}(s, a_v, a_e),\]
and for any $h < H$,
\[\hat N_{hs} =  \argmax_{\nu_h(s) \in \Delta(m)} \min_{\pi_h(s) \in \hat \Pi_{hs}} \E_{\pi_h(s), \nu_h(s)}\brac{R_{he}(s,a_v,a_e) + \sum_{s'} P_h(s' \mid s,a_v,a_e) V^*_{(h+1)e}(s')},\]
to be the set of optimizers of the optimality equations for $V^*_{he}(s)$.

\begin{corollary}\label{cor: solutions}
    For any $s \in S$, $\tilde N_{Hs} = \hat N_{Hs}$ and for any $h < H$,
    \[\tilde N_{hs} = \hat N_{hs} \times \bigtimes_{k>h,s'} \hat N_{ks'}.\]
\end{corollary}

\begin{proof}
    We proceed by backward induction on $h$. For the base case, consider the final time step $h = H$ and fix any $s \in S$. We see that
    \begin{align*}
        \tilde N_{Hs} &= \argmax_{\nu \in N_{Hs}} \min_{\pi \in \tilde \Pi_{Hs}} V_{He}^{\pi, \nu}(s) \\
        &= \argmax_{\nu_h(s) \in \Delta(m)} \min_{\pi_h(s) \in \hat \Pi_{Hs}} V_{He}^{\pi, \nu}(s) \\
        &= \argmax_{\nu_h(s) \in \Delta(m)} \min_{\pi_h(s) \in \hat \Pi_{Hs}} \E_{\pi_h(s),\nu_h(s)} R_{He}(s, a_v,a_e) \\
        &= \hat N_{Hs}.
    \end{align*}
    The second equality used the definition of $N_{Hs}$ and \cref{lem: victim}. The third equality used the bellman-consistency equations~\ref{equ: consistency}.

    For the inductive step, consider any time $h < H$ and fix any $s \in S$. By \cref{lem: recursion}, we know that
    \begin{equation*}
        V^*_{he}(s) = \max_{\nu_h(s) \in \Delta(m)} \min_{\pi_h(s) \in \hat \Pi_{hs}} \E_{\pi_h(s), \nu_h(s)}\brac{R_{he}(s,a_v,a_e) + \sum_{s'} P_h(s' \mid s,a_v,a_e) V^*_{(h+1)e}(s')}.
    \end{equation*}
    Thus, it is clear that $\nu^* \in N_{hs}^*$ if and only if $\nu_h(s)^* \in \hat N_{hs}$ and $\nu^* \in N^*_{(h+1)s'}$ for all reachable $s'$ and $\nu^* \in N_{(h+1)s'}$ is unconstrained for unreachable $s'$. Then, we see that $N^*_{hs} \cap \ \bigcap_{s'} N^*_{(h+1)s'} = \hat N_{hs} \times \ \bigtimes_{s'} N^*_{(h+1)s'}$. Applying these facts to the definition of $\tilde N_{hs}$ gives,
    \begin{align*}
        \tilde N_{hs} &= N_{hs}^* \cap \bigcap_{k > h, s'} N_{ks'}^* \\
        &= ( N^*_{hs} \cap \bigcap_{s'} N^*_{(h+1)s'}) \cap \bigcap_{k > h+1, s'} N_{ks'}^*\\
        &= (\hat N_{hs} \times \bigtimes_{s'} N^*_{(h+1)s'}) \cap \bigcap_{k > h+1, s'} N_{ks'}^* \\
        &= \hat N_{hs} \times \bigtimes_{s'} (N^*_{(h+1)s'} \cap \bigcap_{k > h+1, s'} N_{ks'}^*) \\
        &= \hat N_{hs} \times \bigtimes_{s'} \tilde N_{(h+1)s'} \\
        &= \hat N_{hs} \times \bigtimes_{s'} \bigtimes_{k > h} \tilde N_{ks'} \\
        &= \hat N_{hs} \times \bigtimes_{k > h, s'} \hat N_{ks'}.
    \end{align*}
    The second to last line used the induction hypotheses. 
\end{proof}

\paragraph{Algorithm Correctness.} Now, observe that by definition of $Q^*$, \cref{equ: Q*},
\begin{align*}
    V^*_{he}(s) &= \max_{\nu_h(s) \in \Delta(m)} \min_{\pi_h(s) \in \hat \Pi_{hs}} \E_{\pi_h(s), \nu_h(s)}\brac{R_{he}(s,a_v,a_e) + \sum_{s'} P_h(s' \mid s,a_v,a_e) V^*_{(h+1)e}(s')} \\
    &= \max_{\nu_h(s) \in \Delta(m)} \min_{\pi_h(s) \in \hat \Pi_{hs}} \pi_h(s)^\top Q^*_{he}(s) \nu_h(s) \\
    &= val(LP~\ref{equ: exploiter-LP}(Q^*_{hv}(s),Q^*_{he}(s))).
\end{align*}
Thus, we see at any stage $(s,h)$, $\nu^*_h(s)$ computed from LP~\ref{equ: exploiter-LP} satisfies $\nu^*_h(s) \in \hat N_{hs}$ and so the total partial policy $\nu^*$ computed by \cref{alg: exploiter} satisfies $\nu^* \in \hat N_{hs} \times \bigtimes_{k > h, s'} \hat N_{ks'}$. By \cref{cor: solutions}, we see that $\nu^* \in \tilde N_{hs}$. Thus, the final policy constructed by \cref{alg: exploiter} satisfies $\nu^* \in \bigcap_{h,s} N^*_{hs}$ and so is a \mpsol policy for the exploiter. 

Inductively, we also see using \cref{prop: exploiter-strategy} that $\tilde N_{hs} \neq \varnothing$ and so $\mpsol \neq \varnothing$ completing the proof of \cref{prop: mp=viser}. Applying \cref{prop: exploiter-strategy} then shows $\hat N_{hs}$ has the form that appears in the main text. \cref{cor: solutions} also implies that $\bigcap_{h,s} N^*_{hs} = \bigtimes_{h,s} \hat N_{hs}$, which completes the proof of \cref{prop: exploiter-mstrategy}.

The proof of \cref{thm: mpsol-efficient} is then immediate from the argument in the main text along with \cref{prop: victim-mstrategy} and \cref{prop: exploiter-mstrategy}.

\end{document}